\documentclass{article}
\usepackage[utf8]{inputenc}

\usepackage{amsmath, amsfonts, amsthm, amssymb, wasysym, txfonts, bm, bbm,graphicx, nicefrac}
\usepackage{epstopdf,multicol} 
\usepackage[font=small,labelfont=bf]{caption}
\usepackage{pdfpages, float}

\newtheorem{theorem}{Theorem}[section]
\newtheorem{lem}[theorem]{Lemma}

\theoremstyle{definition}
\newtheorem{defn}{Definition}[section]
\newtheorem{rmk}{Remark}[section]

\makeatletter
\def\blfootnote{\xdef\@thefnmark{}\@footnotetext}
\makeatother

\begin{document}

\title{Explicit optimal-length \\ locally repairable codes of distance 5\blfootnote{This work was supported by the Early Career Research Workshop in Coding Theory, Cryptography, and Number Theory {held at Clemson University in 2018,} under NSF grant DMS-1547399.}}
\author{Allison Beemer\footnote{School of Electrical, Computer and Energy Engineering, Arizona State University,
Tempe, AZ. The author was sponsored by the U.S. Army Research Laboratory under Agreement Number W911NF-17-S-0003.
The views and conclusions contained in this document are those of the authors and should not be interpreted as representing the official policies, either expressed or implied, of the U.S. Army Research Laboratory or the U.S. Government. The U.S. Government is authorized to reproduce and distribute reprints for Government purposes notwithstanding any copyright notation hereon.},
Ryan Coatney\footnote{Department of Mathematics, University of Arizona,
Tucson, AZ.},
Venkatesan Guruswami\footnote{Computer Science Department, Carnegie Mellon University,
Pittsburgh, PA. Research of the author was supported in part by NSF grant CCF-1563742.}, \\
Hiram H. L\'opez\footnote{School of Mathematical and Statistical Sciences, 
Clemson University,
Clemson, SC. The author was partially supported by CONACyT, CVU No. 268999 project ``Network Codes" and Universidad Aut\'onoma de Aguascalientes.}, 
and Fernando Pi\~{n}ero\footnote{Department of Mathematics, University of Puerto Rico at Ponce, Ponce, PR.}}
\date{
}

\maketitle

\begin{abstract} 
Locally repairable codes (LRCs) have received significant recent attention as a method of designing data storage systems robust to server failure. Optimal LRCs offer the ideal trade-off between minimum distance and locality, a measure of the cost of repairing a single codeword symbol. For optimal LRCs with minimum distance greater than or equal to 5, block length is bounded by a polynomial function of alphabet size. In this paper, we give an explicit construction of optimal-length (in terms of alphabet size), optimal LRCs with minimum distance equal to 5.
\end{abstract}

\noindent
{\bf Erratum}\newline
\noindent
In an earlier version we presented a construction of explicit optimal-length locally repairable codes of distance $5$ using cyclic codes, which however was incorrect and only had distance $4$. For that reason the cyclic construction is omitted on this version.

\section{Introduction}

The regular generation of vast amounts of data, and the desire to store this data reliably, serve as the impetus for the design of robust distributed storage systems (DSS). \emph{Locally repairable codes (LRCs)} are a class of codes designed to correct symbol erasures by contacting a small number of other codeword symbols, and have recently attracted a great amount of interest.

We say that an $[n,k,d]_{q}$ linear code is \emph{locally repairable with locality $r$} if each codeword symbol is a function of at most $r$ other symbols. While small locality is desirable, there is a trade-off between locality and the minimum distance of the code, which we simultaneously seek to keep large in the event of many erasures. It was shown in \cite{GHSY12} that an $[n,k,d]_{q}$ code with locality $r$ obeys a Singleton-like bound given by

\[ d\leq n-k+1 - \left(\left\lceil \frac{k}{r} \right\rceil-1\right) = n-k-\left\lceil \frac{k}{r} \right\rceil +2.\]

\noindent We call an LRC which meets this bound \emph{optimal}. 
The Singleton-like bound naturally calls to mind Maximum Distance Separable (MDS) codes, which meet the Singleton bound of $d \leq n-k+1$. In particular, an optimal LRC with $r=k$ is an MDS code.

The MDS conjecture states that there are no non-trivial MDS codes with block length larger than $q+2$, where $q$ is the alphabet size of the code, and that in most cases the upper bound is $q+1$; the case in which $q$ is prime was shown by Ball in \cite{B12}. It is thus natural to speculate as to the relationship between alphabet size and block length for LRCs. Early optimal LRC constructions required alphabet size exponential in block length \cite{HCL07,SRKV13}. In \cite{TB14}, Tamo and Barg used subcodes of Reed-Solomon codes to construct optimal LRCs over alphabet size linear in block length; several other constructions also gave block length $O(q)$ for LRCs with alphabet size $q$ \cite{ZXL17,JMX17,LMX18}. 

Barg et al. then presented constructions of length $\Omega(q^{2})$ in some cases of small distances~\cite{BHHMV17}, distancing the behavior of LRCs from that of MDS codes in this regard. This work was followed closely by the results of \cite{LXY18} which demonstrated optimal LRCs of \emph{unbounded} length for the cases $d=3,4$. Recently, however, Guruswami, Xing, and Yuan~\cite{GXY18} showed that for minimum distance at least $5$, the length of an optimal LRC is in fact bounded by a function of the alphabet size; they also gave simpler constructions of unbounded length for $d=3,4$.

In this paper, we give an explicit construction of optimal LRCs with minimum distance $d=5$ that have largest possible asymptotic length as a function of the alphabet size $q$. In this case, the authors of \cite{GXY18} show that the block length is at most $O(q^{2})$, and also showed a greedy construction to achieve it. {Concurrent work by Jin gives optimal LRCs constructions of length $O(q^{2})$ for minimum distances 5 and 6 via binary constant weight codes \cite{J18}.}
The construction of this paper relies in the Cartesian codes.

Loosely speaking, Cartesian codes are obtained when polynomials with $m$ variables up to a certain total degree are evaluated on a Cartesian set on $m$ components. Each of the $m$ components is a subset of the finite field $\mathbb{F}_q.$ When $m$ is $1,$ Cartesian codes become Reed-Solomon codes. In this paper we will focus on the case when $m$ is $2$ and each component is the multiplicative group $\mathbb{F}_q^{*}=\mathbb{F}_q\setminus\{0\}.$ Cartesian codes were introduced, independently, in \cite{GT13} and \cite{LRV14}. Many properties and applications of Cartesian codes have been studied since their introduction: for example, \cite{CN17} and \cite{BD18} investigate Hamming weights and generalized Hamming weights, respectively, and in \cite{LMM18}, the authors examine the property of being linear complementary dual. 

The paper is organized as follows: in Section \ref{prelims}, we give necessary background and notation. We present optimal LRC constructions using Cartesian codes in Section \ref{Cartesian}. Section \ref{conclusion} concludes the paper.

\section{Preliminaries}
\label{prelims}

We first give several definitions and results that will apply to both code constructions. We begin by formally defining locally repairable codes. Throughout the paper, we will focus on $[n,k,d]_{q}$ linear codes: $k$-dimensional subspaces of $\mathbb{F}_{q}^{n}$ with minimum Hamming distance $d$.
Let $[n]:=\{1,2,\ldots,n\}$.

\begin{defn}
Let $C$ be a $q$-ary block code of length $n$. For each $\alpha\in \mathbb{F}_{q}$ and $i\in [n]$, define

\[C(i,\alpha):= \{ \mathbf{c}=(c_{1},\ldots,c_{n})\in C : c_{i}=\alpha\}.\] 

For a subset $I\subseteq [n]\setminus \{i\}$, we denote by $C_{I}(i,\alpha)$ the projection of $C(i,\alpha)$ onto $I$. For $i\in [n]$, a subset $R$ of $[n]$ that contains $i$ is called a \textit{recovery set} for $i$ if $C_{I_{i}}(i,\alpha)$ and $C_{I_{i}}(i,\beta)$ are disjoint for any $\alpha\neq \beta$, where $I_{i}=R\setminus \{i\}$. Furthermore, $C$ is called a \textit{locally repairable code (LRC) with locality $r$} if, for every $i\in [n]$, there exists a recovery set for $i$ of size $r+1$.
\end{defn}

An \emph{optimal} $[n,k,d]_{q}$ LRC with locality $r$ is an LRC for which equality is met in the Singleton-like bound given by 

\[ d \leq n-k-\left\lceil \frac{k}{r}\right\rceil +2.\]

By Lemma 2.2 and Remark 2 of \cite{GXY18}, achieving equality above is equivalent to the following if $d-2 \not\equiv  r \mod r+1$:

\begin{equation}
\label{optimal_alt}
n-k=\frac{n}{r+1}+d-2-\left\lfloor\frac{d-2}{r+1}\right\rfloor.
\end{equation}

In \cite{GXY18}, the authors show the following for the case $d=5$.

\begin{theorem}[\cite{GXY18}]
\label{upper}
Let $C$ be an optimal $[n, k, d]_q$ locally repairable code of locality $r$, with
$(r+ 1)\mid n$ and parameters satisfying 

\begin{equation}
\label{long_enough}
\frac{n}{r+1}\geq \left( 3-\left\lfloor\frac{3}{r+1}\right\rfloor\right)(3r+2)+\left\lfloor\frac{3}{r+1}\right\rfloor+1.
\end{equation}

\noindent Then, $n=O(q^{2})$.
\end{theorem}

\begin{theorem}[\cite{GXY18}]
\label{lower}
Assume $r\geq 3$ and $(r + 1)\mid n$. Then there exist optimal LRCs of length $\Omega(q^{2})$. In particular, one obtains the best possible length for
optimal LRCs of minimum distance 5.
\end{theorem}

Note that while these results stipulate that $(r+1)\mid n$ for ease of argument, the authors also explain how to extend to the case in which $(r+1)\nmid n$. 

\section{Cartesian Code Construction}
\label{Cartesian}

In this section we describe an optimal distance $5$ LRC over $\mathbb{F}_q$ of length $(q-1)^2$ using
Cartesian codes.
\begin{defn}
Define $V:=\mathbb{F}^*_q \times \mathbb{F}^*_q$ and $n:=\left| V \right|=(q-1)^2.$ 
Fix an ordering $P_1,\ldots, P_{n}$ on the points  of $V.$
Let $\mathcal{F}$ be the $\mathbb{F}_q$-subspace of $\mathbb{F}_q[x,y]$
spanned by the monomials $\left\{x^iy^j \mid 0\leq i,j\leq q-2 \right\}.$
The \emph{Cartesian code} on two components is defined by
\[C(V,\mathcal{F}):=
\left\{ \left(f\left(P_1\right),\ldots,f\left(P_{n}\right) \right)
\mid f\in \mathcal{F} \right\}.\]
\end{defn}
Observe that the vanishing
ideal of $V$ is given by
$I(V)=\left(x^{q-1}-1,y^{q-1}-1\right)$.
Thus the only element of $\mathcal{F}$
that vanishes on $V$ is the zero
element.

For each point $P=(\alpha, \beta)$ of $V$ define the polynomial
\begin{equation}\label{2018.06.30}
g_P(x,y):= \frac{\left(x^{q-1}-1\right)\left(y^{q-1}-1\right)}
{\left(\alpha^{-1}x-1\right)\left(\beta^{-1}y-1\right)}
= \sum\limits_{0 \leq  i,j \leq q-2} \alpha^{-i}\beta^{-j}x^{i}y^{j}.
\end{equation}
It is straightforward to check that $g_P(Q)=1$ if and only if $Q=P$
and $g_P(Q)=0$ if and only if $Q\neq P.$
Furthermore, $g_P(x,y)$ is unique with such a properties, because if there were another polynomial
$f(x,y)$ with same properties, then $f(x,y)-g_P(x,y)$ would vanish on $V,$ but the zero polynomial
is the only polynomial in $\mathcal{F}$ that vanishes on $V.$
\begin{lem}\label{2018.06.28}
Let $f(x,y)$ be an element of $\mathcal{F}.$ Then
\[ f(x,y) = \sum\limits_{P_i \in V} f(P_i)g_{P_i}(x,y), \]
where $f(P_i)$ represents the value of $f(x,y)$ at the point $P_{i}.$
\end{lem}
\begin{proof} Define $h(x,y) : = \sum\limits_{P_i \in V} f(P_i)g_{P_i}(x,y)$. Note that for all $P_i \in V$
we have $f(P_i) = h(P_i)$. Therefore $f(x,y) - h(x,y)$ vanishes on $V$. However, the only polynomial
in $\mathcal{F}$ with this property is the zero polynomial, which implies $f(x,y) = h(x,y).$
\end{proof}
{The following result is key to finding LRCs using Cartesian codes.}
\begin{lem}\label{2018.06.29}
If $P_i=(\alpha_i,\beta_i), i\in [4]$, are distinct elements of $V,$ then the matrix
 \[A=\begin{pmatrix}
1&1&1&1\\
\alpha_1\beta_1&\alpha_2\beta_2&\alpha_3\beta_3&\alpha_4\beta_4\\
\alpha_1^2\beta_1&\alpha_2^2\beta_2&\alpha_3^2\beta_3&\alpha_4^2\beta_4\\
\alpha_1^3\beta_1&\alpha_2^3\beta_2&\alpha_3^3\beta_3&\alpha_4^3\beta_4\\
\alpha_1\beta_1^2&\alpha_2\beta_2^2&\alpha_3\beta_3^2&\alpha_4\beta_4^2\\
\alpha_1\beta_1^3&\alpha_2\beta_2^3&\alpha_3\beta_3^3&\alpha_4\beta_4^3\\
\alpha_1\beta_1^4&\alpha_2\beta_2^4&\alpha_3\beta_3^4&\alpha_4\beta_4^4\\
\end{pmatrix}\]
has linearly independent columns.
\end{lem}
\begin{proof}
If $\beta_1=\beta_2=\beta_3=\beta_4,$ then the $\alpha_i$'s are distinct. 
Thus the first four rows are multiples of the rows of a Vandermonde matrix with $4$ distinct elements,
which means the columns of $A$ are linearly independent. If all the $\beta_i$'s are pairwise distinct,
divide each column $i$ by the element $\alpha_i\beta_i.$
Thus rows $2,5,6$ and $7$ form a Vandermonde matrix with $4$ distinct elements,
so again the columns of $A$ are linearly independent. 

If $\beta_1=\beta_2=\beta_3\neq \beta_4,$
then $\alpha_1,\alpha_2$ and $\alpha_3$ are distinct. Dividing each row by the power
of $\beta_1$ that appears in its first entry, and using row $2$ to simplify rows $5, 6 $ and $7$, $A$ becomes the matrix

\[\begin{pmatrix}
1&1&1&1\\
\alpha_1 &\alpha_2 &\alpha_3 &\alpha_4\gamma\\
\alpha_1^2 &\alpha_2^2 &\alpha_3^2 &\alpha_4^2\gamma\\
\alpha_1^3 &\alpha_2^3 &\alpha_3^3 &\alpha_4^3\gamma\\
0 & 0 & 0 &\alpha_4(\gamma^2-\gamma)\\
0 & 0 & 0 &\alpha_4(\gamma^3-\gamma)\\
0 & 0 & 0 &\alpha_4(\gamma^4-\gamma)\\
\end{pmatrix},\]

\noindent where  $\gamma:=\beta_4\beta_1^{-1}\neq 1.$ Then row $4$ is non-zero and all the columns
are linearly independent, because columns $1,2$ and $3$ contain a Vandermonde matrix with
$3$ distinct elements, and the first three columns cannot span the fourth. 

If $\beta_1=\beta_2$ and $\beta_2\neq\beta_3\neq \beta_4,$
then $\alpha_1$ and $\alpha_2$ are different. For $i\in [4]$, divide column $i$ of $A$ by $\alpha_i$ and subtract the
second column from the first. Matrix $A$ becomes the following:
 \[\begin{pmatrix}
 \alpha_1^{-1}-\alpha_2^{-1}&\alpha_2^{-1}&\alpha_3^{-1}&\alpha_4^{-1}\\
0&\beta_2&\beta_3&\beta_4\\
(\alpha_1-\alpha_2)\beta_1&\alpha_2\beta_2&\alpha_3\beta_3&\alpha_4\beta_4\\
(\alpha_1^2-\alpha_2^2)\beta_1&\alpha_2^2\beta_2&\alpha_3^2\beta_3&\alpha_4^2\beta_4\\
0&\beta_2^2&\beta_3^2&\beta_4^2\\
0&\beta_2^3&\beta_3^3&\beta_4^3\\
0&\beta_2^4&\beta_3^4&\beta_4^4\\
\end{pmatrix}.\]
As $\alpha_1\neq \alpha_2,$ position $(3,1)$ is non-zero. Thus, the columns of $A$ are linearly independent
because columns $2, 3$ and $4$ contain a Vandermonde matrix with $3$ distinct elements, and these three columns cannot span the first.

Finally, if $\beta_1 = \beta_2 \neq \beta_3=\beta_4,$ then $\alpha_1\neq \alpha_2$ and
$\alpha_3\neq \alpha_4.$
Subtract the first column from the second one, the third column from the fourth, and
then the first column from the third.
Divide each row by the power of $\beta_1$ that appears in its first entry;
divide the second column by $\alpha_2-\alpha_1$ and the fourth column by
$(\alpha_4-\alpha_3)\gamma,$ where $\gamma = \beta_3 \beta_1^{-1}.$ Then $A$ becomes:

 \[\begin{pmatrix}
1&0&0&0\\
\alpha_1 &1 &\alpha_3\gamma -\alpha_1& 1 \\
\alpha_1^2 &(\alpha_2+\alpha_1) &\alpha_3^2\gamma-\alpha_1^2& (\alpha_4+\alpha_3)\\
\alpha_1^3 &(\alpha_2^2+\alpha_2\alpha_1+\alpha_1^2) &\alpha_3^3\gamma-\alpha_1^3&(\alpha_4^2+\alpha_4\alpha_3+\alpha_3^2)\\
\alpha_1&1 &\alpha_3\gamma^2-\alpha_1& \gamma\\
\alpha_1& 1&\alpha_3\gamma^3-\alpha_1& \gamma^2\\
\alpha_1&  1&\alpha_3\gamma^4-\alpha_1&\gamma^3\\
\end{pmatrix}.\]

Subtracting row $2$ from rows $5,6$ and $7$ we obtain:

  \[A=\begin{pmatrix}
1&0&0&0\\
\alpha_1 &1 &\alpha_3\gamma -\alpha_1& 1 \\
\alpha_1^2 &(\alpha_2+\alpha_1) &\alpha_3^2\gamma-\alpha_1^2& (\alpha_4+\alpha_3)\\
\alpha_1^3 &(\alpha_2^2+\alpha_2\alpha_1+\alpha_1^2) &\alpha_3^3\gamma-\alpha_1^3&(\alpha_4^2+\alpha_4\alpha_3+\alpha_3^2)\\
0& 0&\alpha_3(\gamma^2-\gamma)& \gamma-1\\
0& 0&\alpha_3(\gamma^3-\gamma)& \gamma^2-1\\
0& 0&\alpha_3(\gamma^4-\gamma)&\gamma^3-1\\
\end{pmatrix}.\]

It is clear that columns 1 and 2 are independent. As $\gamma\neq 1,$ position $(5,4)$ is non-zero, thus columns $1,2$ and $4$ are linearly
independent. If columns 1, 2, and 4 span column 3, then column $3$ must be a multiple of column $4.$ Since $\gamma\neq 1$, rows 5, 6, and 7 imply this multiple must be $\alpha_{3}\gamma$. However, this would imply that $\alpha_{1}=0$ in the second row, a contradiction. Thus, all columns are independent.
\end{proof}

\begin{theorem}
\label{cart_const_thm}
Assume $(r+1)\mid (q-1)$.
Let $\mathcal{L}$ be the $\mathbb{F}_q$-subspace of $\mathbb{F}_q[x,y]$
spanned by the monomials

\[
\left\{ x^iy^j \ | \ 0 \leq i,j \leq q-2, \ i \not\equiv r \mod r+1  \right\} \setminus 
\left\{ 1, x^{q-3}y^{q-2}, x^{q-4}y^{q-2} \right\}.
\]

The code $C(V,\mathcal{L})$ is an LRC with locality $r$ and minimum distance $\geq 5.$
If $r> 3$, $C(V,\mathcal{L})$ is an optimal LRC with locality $r$ and
minimum distance $5.$
 \end{theorem}
\begin{proof}
Consider the sets of the form $\alpha R \times \left\{\gamma\right\}
:= \left\{\left(\alpha \beta, \gamma\right) \mid \beta \in 
R\right\}\subset V$, where
$R := \left\{ \beta \in \mathbb{F}_q^* \mid \beta^{r+1} =1\right\}$. These partition $V$ into $\frac{n}{r+1}$ disjoint sets, each of size $r+1$; we claim these are recovery sets of the code.
Evaluating all monomials of $\mathcal{L}$ at $\alpha R \times \{ \gamma \}$ reduces to polynomials
in $\mathbb{F}_q[x]$ of degree less than $r.$ Therefore there is a single parity check equation for the $r+1$ points in each set $\alpha R \times \{\gamma\}$, and so these form recovery sets.

Now we prove that the minimum distance of $C(V,\mathcal{L})$ is at least $5.$
Let $f(x,y)$ be an element of $\mathcal{L}.$ If $f(x,y)$ has only $4$ non-zero evaluations on $V,$ by Lemma~\ref{2018.06.28} there are $P_\ell=(\alpha_\ell,\beta_\ell)\in V$
and $a_\ell\in \mathbb{F}_q^*,$ for $\ell\in [4]$ such that
\[f(x,y) = a_1g_{P_1}(x,y) + a_2g_{P_2}(x,y)+a_3g_{P_3}(x,y)+a_4g_{P_4}(x,y).\]
By \eqref{2018.06.30}, the coefficient of the monomial $x^iy^j$ in $g_{P_i}(x,y)$ is
given by $\alpha_\ell^{-i}\beta_\ell^{-j}.$ As the monomial $1$ does not belong to the support 
of $f(x,y),$ then $a_1+a_2+a_3+a_4=0.$ As $x^{q-2}y^{q-2}$ does not belong to the support
of $f(x,y),$ and $\alpha_\ell^{-q+2}\beta_\ell^{-q+2}=\alpha_\ell\beta_\ell,$ then
$a_1\alpha_1\beta_1+a_2\alpha_2\beta_2+a_3\alpha_3\beta_3+a_4\alpha_4\beta_4=0.$
In a similar way, as the monomials $x^{q-3}y^{q-2}$, $x^{q-4}y^{q-2}$,
$x^{q-2}y^{q-3}$, $x^{q-2}y^{q-4}$, $x^{q-2}y^{q-5}$
do not belong to the support of $f(x,y)$, and its coefficients are of the form
$\alpha_\ell^2\beta_\ell,$ $\alpha_\ell^3\beta_\ell,$ $\alpha_\ell\beta_\ell^2,$ $\alpha_\ell\beta_\ell^3,$
$\alpha_\ell\beta_\ell^4,$ then the
following matrix has linearly dependent columns

 \[A=\begin{pmatrix}
1&1&1&1\\
\alpha_1\beta_1&\alpha_2\beta_2&\alpha_3\beta_3&\alpha_4\beta_4\\
\alpha_1^2\beta_1&\alpha_2^2\beta_2&\alpha_3^2\beta_3&\alpha_4^2\beta_4\\
\alpha_1^3\beta_1&\alpha_2^3\beta_2&\alpha_3^3\beta_3&\alpha_4^3\beta_4\\
\alpha_1\beta_1^2&\alpha_2\beta_2^2&\alpha_3\beta_3^2&\alpha_4\beta_4^2\\
\alpha_1\beta_1^3&\alpha_2\beta_2^3&\alpha_3\beta_3^3&\alpha_4\beta_4^3\\
\alpha_1\beta_1^4&\alpha_2\beta_2^4&\alpha_3\beta_3^4&\alpha_4\beta_4^4\\
\end{pmatrix}.\]

This is not possible, by Lemma~\ref{2018.06.29}. Using same matrix we conclude $f(x,y)$ 
cannot have only $3$ non-zero elements in $V$, because in such a case, matrix $A$ would have $3$ linearly dependent columns. In addition,
$f(x,y)$ cannot have only $2$ non-zero elements in $V$, because in such a case, matrix $A$ would have $2$ linearly dependent columns. 
Finally $f(x,y)$ cannot have only one
non-zero element in $V$ because by \eqref{2018.06.30}, its support should contain all the monomials. Thus, the minimum distance of the code is at least 5.

Recall that $\dim C(V, \mathcal{L}) = \# \text{Mon}(\mathcal{L}) = n - \dfrac{n}{r+1}-3.$
Thus, when $r > 3$, by the Singleton-like bound the minimum distance of $C(V,\mathcal{L})$ is at most $5$. Since we have shown that the minimum distance is at least $5$, 
$C(V,\mathcal{L})$ is an optimal LRC when $r > 3$.
\end{proof}

\begin{rmk}
Requiring $\lceil \nicefrac{(q-1)}{(r+1)} \rceil>3$ will be sufficient to guarantee that the LRC of Theorem \ref{cart_const_thm} is of optimal length, while if $\lceil \nicefrac{(q-1)}{(r+1)} \rceil\leq 3$, $r$ and $q$ must be large in order to satisfy the inequality in Equation \ref{long_enough}.
\end{rmk}

\section{Conclusions}
\label{conclusion}

In this paper, we used Cartesian codes to construct a family of optimal-length, optimal LRCs for the case in which the minimum distance is equal to 5. Ongoing work includes extending our arguments to higher minimum distance and exhibiting other algebraic constructions of optimal-length LRCs.


\bibliographystyle{IEEEtran}
\bibliography{LRC_bib}

\end{document}